\documentclass{article}
\usepackage{graphicx}
\usepackage{amsfonts}    
\usepackage{amssymb}     
\usepackage{amsmath}     
\usepackage{color}
\usepackage{amsthm}

\newtheorem{definition}{Definition}
\newtheorem{proposition}{Proposition}

\newtheorem{example}{Example}
\newtheorem{algorithm}{Algorithm}
\newtheorem{lemma}{Lemma}
\newtheorem{remark}{Remark}

\newcommand{\rou}[1]{\textcolor{red}{#1}}

\newcommand{\rrbr}{]\!]}
\newcommand{\llbr}{[\![}

\begin{document}

\title{Generalized distance to a simplex and a new geometrical method for portfolio optimization}

\author{Fr\'{e}d\'{e}ric BUTIN\footnote{Universit\'e de Lyon, Universit\'e Lyon 1, CNRS, UMR5208, Institut Camille Jordan, 43 blvd du 11 novembre 1918, F-69622 Villeurbanne-Cedex, France, butin@math.univ-lyon1.fr}}

\maketitle

\begin{abstract}
\noindent Risk aversion plays a significant and central role in investors' decisions in the process of developing a portfolio. In this framework of portfolio optimization we determine the portfolio that possesses the minimal risk by using a new geometrical method. For this purpose, we elaborate an algorithm that enables us to compute any generalized Euclidean distance to a standard simplex. With this new approach, we are able to treat the case of portfolio optimization without short-selling in its entirety, and we also recover in geometrical terms the well-known results on portfolio optimization with allowed short-selling.\\
Then, we apply our results in order to determine which convex combination of the CAC 40 stocks possesses the lowest risk: not only we get a very low risk compared to the index, but we also get a return rate that is almost three times better than the one of the index.\\


\end{abstract}

\section{Introduction and aims of the article}\label{section1}

\subsection{\textbf{Framework}}

\noindent The paper~\cite{M52} published by Harry Markowitz in 1952 completely changed the methods of portfolio management and gave birth to the so-called ``Modern Portfolio Theory'', thanks to which its author earned the Nobel Prize in Economics in 1990. Since his works and the paper~\cite{S63} of Sharpe, this theme centralizes a lot of interest and many developments have been written in this domain. Let us recall some recent and important works to which our article is linked.\\

\noindent In~\cite{DJDY08}, J\'{o}n Dan\'{\i}elsson, Bj{\o}rn N. Jorgensen, Casper G. de Vries and Xiaoguang Yang study the portfolio allocation under the probabilistic VaR constraint and obtain remarkable topological results: the set of feasible portfolios is not always connected nor convex, and the number of local optima increases in an exponential way with
the number of states. They propose a solution to reduce computational complexity due to this exponential increase.\\

\noindent In~\cite{FS12}, Claudio Fontana and Martin Schweizer give a simple approach to mean-variance portfolio problems: they change the problems' parametrisation from trading strategies to final positions. In this way they are able to solve many quadratic optimisation problems by using orthogonality techniques in Hilbert spaces and providing explicit formulas.\\

\noindent In their important article~\cite{BCGD18}, Hanene Ben Salah, Mohamed Chaouch, Ali Gannoun and Christian De Peretti (see also the thesis~\cite{B15}) define a new portfolio optimization model in which the risks are measured thanks to conditional variance or semivariance. They use returns prediction obtained by nonparametric
univariate methods to make a dynamical portfolio selection and get better performance.\\

\noindent In~\cite{PR19}, Sarah Perrin and Thierry Roncalli show how four algorithms of optimization (the coordinate descent, the alternating direction method of multipliers, the proximal gradient method and the Dykstra's algorithm) can be used to solve problems of portfolio allocation.\\

\noindent In \cite{BIPS20}, Taras Bodnar, Dmytro Ivasiuk, Nestor Parolya and Wolfgang Schmid make an interesting work about the portfolio choice problem for power and logarithmic utilities: they compute the portfolio weights for these utility functions assuming that the portfolio returns follow an approximate log-normal distribution, as suggested in \cite{B00}. It is also noticeable that their optimal portfolios belong to the set of mean-variance feasible portfolios.

\subsection{\textbf{Aims and organization of the paper}}

\noindent The three aims of this article are the following:\\
--- give a new geometrical algorithm (Algorithm~\ref{alg}) to compute any generalized distance to a simplex,\\
--- determine, by making use of this algorithm, a portfolio with minimal variance,\\
--- apply this technique to the CAC 40 stocks, and get a portfolio with return rate that is almost three times better than the one of the index.\\

\noindent After having briefly explained the notations in \textbf{Section~\ref{section1}}, we expose in \textbf{Section~\ref{section2}} the portfolio optimization problem and prove by compactness and convexity arguments that it possesses a unique solution.\\
Then, in \textbf{Section~\ref{section3}}, we solve the problem in the case where short-selling is allowed. For this, we recall the classical method, and we give our very simple geometrical method.\\
\textbf{Section~\ref{section4}} is the heart of the article: in this section, we solve the portfolio optimization problem in the case where short-selling is not allowed. For this purpose, we give a new geometrical algorithm to compute the distance from a point to a standard simplex, which can be used for every Euclidean distance.\\
We can eventually apply this algorithm to the example of the CAC 40 stocks and determine the portfolio with the lowest risk. This portfolio also has the property of being almost three times more efficient than the underlying index. This is done in \textbf{Section~\ref{section5}}.

\subsection{\textbf{Notations}}

\noindent We consider $n$ stocks $S_1,\,S_2,\dots,\,S_n$ and denote by $X_1,\,X_2,\dots,\,X_n$ the random variables that represent their return rate (for example, daily, monthly or yearly). For every $i\in\llbr1,\,n\rrbr$, we set $m_i=E(X_i)$ (mean of $X_i$) and $V_i=V(X_i)$ (variance of $X_i$). We set $$C=(Cov(X_i,\,X_j))_{(i,\,j)\in\llbr1,\,n\rrbr^2}=\begin{bmatrix}
Cov(X_1,\,X_1) & \cdots & Cov(X_1,\,X_n)\\
\vdots &  & \vdots\\
Cov(X_n,\,X_1) & \cdots & Cov(X_n,\,X_n)
\end{bmatrix},$$
$$X=\begin{bmatrix}
X_1\\
\vdots\\
X_n
\end{bmatrix},\ \ \textrm{and}\ \ m=\begin{bmatrix}
m_1\\
\vdots\\
m_n
\end{bmatrix}.$$
Matrix $C$ is the \emph{covariance matrix} of $X_1,\,X_2,\dots,\,X_n$, and $X$ is a random vector.\\

\begin{definition}
We call \emph{portfolio} (with allowed short-selling) every linear combination $\displaystyle{P_\mathbf{x}=\sum_{i=1}^nx_iS_i}$, where $\mathbf{x}=(x_1,\dots,\,x_n)\in\mathbb{R}^n$ and $\displaystyle{\sum_{i=1}^nx_i=1}$.\\
The \emph{return rate} of the portfolio is the linear combination $\displaystyle{R_\mathbf{x}=\sum_{i=1}^nx_iX_i}$.
\end{definition}

\noindent If we don't allow short-selling, then every $x_i$ must be nonnegative, and in that case the linear combination is a \emph{convex combination}.

\section{Minimisation of the variance of a portfolio}\label{section2}

\subsection{\textbf{The standard simplex of dimension $n-1$}}

\noindent Let us recall some classical results that we will use in the following.

\begin{proposition}\label{esppvar}
We have $E(R_\mathbf{x})=\,^t\mathbf{x}m$ and $V(R_\mathbf{x})=\,^t\mathbf{x}C\mathbf{x}$.
\end{proposition}

\begin{proof}
We immediately have $\displaystyle{E(R_\mathbf{x})=\sum_{i=1}^nx_iE(X_i)=\,^t\mathbf{x}m}$.\\
Moreover,
$$\begin{array}{lll}
V(R_\mathbf{x}) & = & \displaystyle{E\left(\left(\sum_{i=1}^nx_iX_i-E\left(\sum_{i=1}^nx_iE(X_i)\right)\right)^2\right)}\\
 & = & \displaystyle{E\left(\left(\sum_{i=1}^nx_i(X_i-E(X_i))\right)^2\right),}
\end{array}$$
hence
$$\begin{array}{lll}
V(R_\mathbf{x}) & = & \displaystyle{E\left(\sum_{i=1}^n\sum_{j=1}^nx_ix_j(X_i-E(X_i))(X_j-E(X_j))\right)}\\
 & = & \displaystyle{\sum_{i=1}^n\sum_{j=1}^nx_ix_jCov(X_i,\,X_j).}
\end{array}$$
\end{proof}

\noindent The following proposition is immediate.

\begin{proposition}\label{matrC}
The matrix $C$ has the two following properties.
\begin{description}
  \item[(i)] It is symmetric positive,
  \item[(ii)] It is symmetric definite positive if and only if $X_1,\dots,\,X_n$ are almost surely affinely independent.
\end{description}
\end{proposition}

\noindent In all the following, we assume that $C$ is symmetric \emph{definite} positive: this is always true in practise.

\noindent Let us denote by $H$ the affine hyperplane of $\mathbb{R}^n$ with equation $\displaystyle{\sum_{j=1}^nx_j=1}$, and by $K$ the standard $(n-1)-$simplex (also called standard simplex of dimension $n-1$), i.e.
$$K:=\left\{\mathbf{x}=(x_1,\dots,\,x_n)\in[0,\,1]^n\ /\ \sum_{j=1}^nx_j=1\right\}.$$
The simplex $K$ is a Haussdorff compact subset of $\mathbb{R}^n$ that is contained in the hyperplane $H$.\\
For example, if $n=2$, $H$ is the line of equation $x+y=1$ in the usual plane, and the $1-$standard simplex $K$ is the segment $[(0,\,1),\,(1,\,0)]$ of this line.

\subsection{\textbf{Existence and uniqueness of the solution}}

\noindent Minimizing the variance of the portfolio is equivalent to finding the minimum on $K$ of the quadratic form
$$f:\mathbf{x}\mapsto V(R_\mathbf{x})=\,^t\mathbf{x} C\mathbf{x}.$$
Let us consider the scalar product
$$(\mathbf{x},\,\mathbf{y})\mapsto\langle \mathbf{x},\,\mathbf{y}\rangle:=\,^t\mathbf{x} C\mathbf{}y$$
and the Euclidean norm
$$\mathbf{x}\mapsto\|\mathbf{x}\|:=\sqrt{\,^t\mathbf{x} C\mathbf{x}}.$$
The aim is to determine the point of $K$ that realizes the minimal distance to $K$ from the origin point in the sense of~$\|\cdot\|$.\\
As $K$ is a Haussdorff compact subset, and as $f$ is continuous, we know that this minimum does exist. We will prove that it is also unique. For this, the convexity plays a central role (see for example \cite{R17}).

\begin{proposition}\label{fconv}
The map $f$ is strictly convex.
\end{proposition}

\begin{proof}
For every $x,\,y\in\mathbb{R}^n$ and for every $\lambda\in[0,\,1]$, we have
$$f(\lambda x+(1-\lambda)y)=\lambda^2\,^txCx+(1-\lambda)^2\,^tyCy+\lambda(\lambda-1)\left(\,^txCy+\,^tyCx\right),$$
hence
$$\lambda f(x)+(1-\lambda)f(y)-f(\lambda x+(1-\lambda)y)=\lambda(1-\lambda)\,^t(x-y)C(x-y),$$
which is nonnegative, as $C$ is positive.\\
Moreover, if $x\neq y$, this quantity is positive as $C$ is definite positive.
\end{proof}


\begin{proposition}\label{minconv}
Let $K$ be a convex domain and $f:K\rightarrow\mathbb{R}$ a convex map. Then,
\begin{description}
  \item[(i)] every local minimum of $f$ is global,
  \item[(ii)] if $f$ is strictly convex, then $f$ possesses at most one minimum.
\end{description}
\end{proposition}

\begin{proof}
(i) Let $x_0$ be a local minimum of $f$ in $K$. Let us assume by contradiction that it is not global: there exists $y\in K$ such that $f(y)<f(x_0)$. For every $t\in]0,\,1[$, let us set $y_t=ty+(1-t)x_0$. Then $y_t$ belongs to $K$, and for $t$ small enough, $\|y_t-x_0\|=t\|y-x_0\|$ is close enough to $0$, i.e. $y_t$ is close enough to $x_0$. Thus $f(x_0)\leq f(y_t)$. As $f$ is convex, we have $f(y_t)\leq tf(y)+(1-t)f(x_0)$, hence $f(x_0)\leq tf(y)+(1-t)f(x_0)$, i.e. $f(x_0)\leq f(y)$, which is a contradiction.\\
(ii) Let us assume by contradiction that $f$ possesses a least two distinct minima $x_1$ and $x_2$ with $f(x_1)=f(x_2)$. Then, as $f$ is strictly convex, $f\left(\frac{x_1+x_2}{2}\right)<\frac{1}{2}f(x_1)+\frac{1}{2}f(x_2)=f(x_1)$, which is a contradiction to the fact that $x_1$ is a minimum.
\end{proof}

\noindent As a consequence of Proposition~\ref{minconv}, the quadratic form $f$ possesses exactly one minimum on $f$, and his minimum is global.

\section{Minimisation of $f$ on $H$}\label{section3}

\noindent In this section, we give two methods to compute the portfolio that possesses the lowest risk: the classical one, and our geometrical approach.\\

\noindent Let us denote by $\mathcal{E}=(e_1,\dots,\,e_n)$ the canonical basis of $\mathbb{R}^n$. A basis of the vector hyperplane that directs $H$ is $\mathcal{B}=(e_1-e_2,\,e_1-e_3,\dots,\,e_1-e_n)$. Moreover, for every $\displaystyle{\mathbf{x}=\sum_{j=1}^nx_je_j\in \mathbb{R}^n}$, $\mathbf{x}$ belongs to $H$ if and only if $\displaystyle{\sum_{j=1}^nx_j=1}$.\\

\noindent Let is denote by $u$ the vector $u=\displaystyle{\sum_{j=1}^ne_j}$, which is orthogonal to the hyperplane~$H$, and let us set
$$h:\mathbf{x}=(x_1,\dots,\,x_n)\mapsto\,^t\mathbf{x}u-1=\sum_{j=1}^nx_j-1.$$

\subsection{\textbf{Minimisation of $f$ on $H$ by the classical method}}

\noindent Here we briefly recall the classical method to compute the portfolio that possesses the lowest risk. Several sources, such as \cite{N09}, \cite{M10} and \cite{PP14}, provide a clear presentation of these well-known tools.

\begin{proposition}\label{methclass}
The unique solution $\mathbf{x}_0$ that minimises the map $f$ on $H$ is the vector $\displaystyle{\mathbf{x}_0=\frac{C^{-1}u}{\,^tuC^{-1}u}}$.
\end{proposition}

\begin{proof}
According to Lagrange's multipliers theorem, there exists $\lambda\in\mathbb{R}$ such that $\nabla f(x_0)=\lambda\nabla h(x_0)$. Here we have $\nabla f(x)=2Cx$ and $\nabla h(x)=u$, thus $2Cx_0=\lambda u$, i.e. $x_0=\frac{\lambda}{2}C^{-1}u$. As $\,^tx_0u=1$, we get
$$1=\,^tx_0u=\,^tux_0=\frac{\lambda}{2}\,^tuC^{-1}u,$$
hence $\displaystyle{\lambda=\frac{2}{\,^tuC^{-1}u}}$ and $\displaystyle{x_0=\frac{C^{-1}u}{\,^tuC^{-1}u}}$.
\end{proof}

\begin{example}\label{methclass2}
For $n=2$, the formula is very simple: by setting
$$\Delta=V(X_1)V(X_2)-Cov(X_1,\,X_2)^2,$$
we have
$$C=\begin{bmatrix}
V(X_1) & Cov(X_1,\,X_2)\\
Cov(X_1,\,X_2) & V(X_2)
\end{bmatrix}$$
and
$$C^{-1}=\frac{1}{\Delta}\begin{bmatrix}
V(X_2) & -Cov(X_1,\,X_2)\\
-Cov(X_1,\,X_2) & V(X_1)
\end{bmatrix},$$
hence
$$\displaystyle{x_0=\frac{V(X_2)-Cov(X_1,\,X_2)}{V(X_2)-2Cov(X_1,\,X_2)+V(X_1)}e_1+\frac{V(X_1)-Cov(X_1,\,X_2)}{V(X_2)-2Cov(X_1,\,X_2)+V(X_1)}e_2.}$$
\end{example}

\subsection{\textbf{Minimisation of $f$ on $H$ by the geometrical approach}}

\noindent Here we recover the classical results on the portfolio with minimal variance with allowed short-selling by making use of an Euclidean interpretation.\\
This portfolio is $P_{\mathbf{x}_0}$, where $\mathbf{x}_0$ is the orthogonal projection onto $H$ of the origin point.\\

\noindent Let us define the $(n,\,n)-$matrix
$$A=\begin{bmatrix}
c_{1,1}-c_{1,2} & c_{1,2}-c_{2,2} & \cdots & c_{1,n}-c_{n,2}\\
c_{1,1}-c_{1,3} & c_{1,2}-c_{2,3} & \cdots & c_{1,n}-c_{n,3}\\
\vdots & \vdots & & \vdots\\
c_{1,1}-c_{1,n} & c_{1,2}-c_{2,n} & \cdots & c_{1,n}-c_{n,n}\\
1 & 1 & \cdots & 1
\end{bmatrix}.$$

\begin{proposition}\label{propmatrices}
The unique solution $\mathbf{x}_0$ that minimises $f$ on $H$ is the vector whose coordinates in $\mathcal{E}$ are given by the last column of the inverse of matrix~$A$.
\end{proposition}

\begin{proof}
For every $\mathbf{x}$ in $H$, $\mathbf{x}$ is the orthogonal projection onto $H$ of the origin point if and only if $\mathbf{x}$ is orthogonal to $H$, that is to say, for every $i\in\llbr2,\,n\rrbr$, $\langle \mathbf{x},\,e_1-e_i\rangle=0$.\\
Since $\displaystyle{\langle \mathbf{x},\,e_1-e_i\rangle=\sum_{j=1}^nx_j\langle e_j,\,e_1-e_i\rangle=\sum_{j=1}^nx_j(c_{1,j}-c_{i,j})}$, we deduce that $\mathbf{x}$ is the solution if and only if $A\mathbf{x}=\begin{bmatrix}
\mathbf{0}_{n-1,1}\\
1
\end{bmatrix}$, i.e. $\mathbf{x}=A^{-1}\begin{bmatrix}
\mathbf{0}_{n-1,1}\\
1
\end{bmatrix}$, which means that the coordinates in $\mathcal{E}$ of $\mathbf{x}$ are given by the last column of $A^{-1}$.
\end{proof}

\section{Distance to $K$ from a point of $\mathbb{R}^n$}\label{section4}

\noindent In this section, we will solve the problem of portfolio optimization without short-selling, by giving an explicit and calculable solution that doesn't seem to appear in the literature.

\subsection{\textbf{Projections onto affine hyperplanes of $\mathbb{R}^m$}}

\noindent Let us now consider $J$ a subset of $\llbr1,\,n\rrbr$ and the vector subspace $\displaystyle{E=\bigoplus_{j\in J}\mathbb{R}e_j}$ of $\mathbb{R}^n$, identified with $\mathbb{R}^m$, where $m=|J|$. Let $H'$ be the affine hyperplane of $E$ defined by the equation $\displaystyle{\sum_{j\in J}x_j=0}$. Let us fix $i_0\in J$ and define $J_1$ by $J_1=J\setminus\{i_0\}$, and let us denote by $\overline{J}$ the complementary of $J$ in $\llbr1,\,n\rrbr$. Then, a basis of the vector hyperplane of $E$ parallel to $H'$ is $\mathcal{B}'=(e_{i_0}-e_i\ /\ i\in J_1)$.\\
Let $\mathbf{a}$ be a point of $\mathbb{R}^n$.\\

\noindent Let us set $$B'=\left(c_{i,j}-c_{i_0,j}\right)_{(i,j)\in J_1\times J}\ \textrm{and}\ B=\begin{bmatrix}
B'\\
\mathbf{1}_{1,m}
\end{bmatrix},$$ then $$\displaystyle{b'=\left(\sum_{j=1}^n a_j(c_{i,j}-c_{i_0,j})\right)_{i\in J_1}}\ \textrm{and}\ b=\begin{bmatrix}
b'\\
1
\end{bmatrix}.$$

\begin{proposition}\label{vp}
The orthogonal projection of $\mathbf{a}$ onto $H'$ is the vector whose nonzero coordinates in $\mathcal{E}$ are given by $B^{-1}b$, which means that $(x_i)_{i\in J}=B^{-1}b$ and $(x_i)_{i\in \overline{J}}=0_{n-m,1}$.
\end{proposition}

\begin{proof}
For every $\displaystyle{\mathbf{x}=\sum_{j=1}^nx_je_j}\in\mathbb{R}^n$, $\mathbf{x}$ is the orthogonal projection of $\mathbf{a}$ onto~$H'$ if and only if the three following conditions hold
\begin{enumerate}
  \item[(i)] for every $j\in\overline{J}$, $x_j=0$,
  \item[(ii)] $\displaystyle{\sum_{j=1}^nx_j=1}$,
  \item[(iii)] for every $i\in J_1$, $\mathbf{x}-\mathbf{a}$ is orthogonal to $e_{i_0}-e_i$.
\end{enumerate}
Since $$\displaystyle{\langle\mathbf{x}-\mathbf{a},\,e_{i_0}-e_j\rangle=\sum_{j=1}^n(x_j-a_j)\langle e_j,\,e_{i_0}-e_i\rangle},$$
we have $\langle\mathbf{x}-\mathbf{a},\,e_{i_0}-e_j\rangle=0$ if and only if
$$\displaystyle{\sum_{j\in J}x_j\langle e_i-e_{i_0},\,e_j\rangle=\sum_{j=1}^na_j\langle e_i-e_{i_0},\,e_j\rangle},$$
i.e. $\displaystyle{\sum_{j\in J}x_j(c_{i,j}-c_{i_0,j})=\sum_{j=1}^na_j(c_{i,j}-c_{i_0,j})}$.
\end{proof}

\subsection{\textbf{The algorithm for computing the generalized distance to $K$ from a point of~$\mathbb{R}^n$}}

\noindent We now propose a recursive algorithm to compute the point $\mathbf{x}_0$ realizing the distance to $K$ from a point $\mathbf{a}\in\mathbb{R}^n$. In his article~\cite{C16}, L.~Condat gave a new and fast algorithm to project a vector onto a simplex. However, his algorithm was made only for the usual Euclidean distance. Our algorithm can be used for every Euclidean distance. The reader can also have a look at the paper~\cite{CY11} about the projection onto a simplex.

\begin{algorithm}\label{alg}$\\$
\verb"Entry: ("$\mathbf{a}$\verb","$K$\verb")"\\
\verb"Compute" $\mathbf{x}$ \verb"the orthogonal projection of" $\mathbf{a}$ \verb"onto" $H$\\
\verb"If" $\mathbf{x}$ \verb"belongs to" $K$\\
\verb"    Return" $\mathbf{x}$\\
\verb"Else"\\
\verb"    If" $K$ \verb"is a" $1-$\verb"simplex (i.e." $K$ \verb"has exactly" $2$ \verb"vertices)"\\
\verb"        Return the vertex that is the closest to" $\mathbf{x}$\\
\verb"    Else"\\
\verb"        Determine the hyperface" $K'$ \verb"of" $K$ \verb"that is the closest to" $\mathbf{x}$\\
\verb"        Compute" $\mathbf{y}$ \verb"the orthogonal projection of" $\mathbf{x}$ \verb"onto" $H'$ \verb"(the affine"\\
\verb"        subspace defined by" $K'$\verb")"\\
\verb"        Apply recursively the algorithm to ("$\mathbf{y}$\verb","$K'$\verb")"\\
\end{algorithm}

\begin{figure}
\includegraphics[height=4cm]{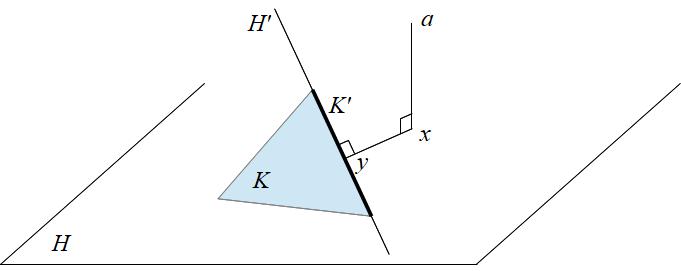}
\caption{Geometric explanation of Algorithm~\ref{alg}}
\label{fig:1}
\end{figure}

\begin{proposition}\label{termin}
Algorithm~\ref{alg} ends.
\end{proposition}

\begin{proof}
This is straightforward since at each step of the algorithm the dimension of the simplex decreases of one unit.
\end{proof}

\begin{lemma}\label{lemme}
If $\mathbf{x}$ belongs to $H\setminus K$, then the distance from $\mathbf{x}$ to $K$ is realized in a point of the frontier of $K$.
\end{lemma}

\begin{proof}
Let us proceed by contradiction by assuming that the distance from $\mathbf{x}$ to $K$ is realized in a point $\mathbf{z}$ in $\stackrel{\circ}{K}$. Let us denote by $\mathbf{y}$ the intersection of the line $(\mathbf{x},\,\mathbf{z})$ with an hyperface of $K$ crossed by this line. Then, by Minkowski, we get $\|\mathbf{x}-\mathbf{z}\|=\|\mathbf{x}-\mathbf{y}\|+\|\mathbf{y}-\mathbf{z}\|>\|\mathbf{x}-\mathbf{y}\|$, which is absurd since $\mathbf{z}$ realizes the minimal distance from $\mathbf{x}$ to $K$.
\end{proof}

\noindent As a consequence, if $\mathbf{x}$ belongs to $H\setminus K$, then the distance from $\mathbf{x}$ to $K$ is the distance from $\mathbf{x}$ to the hyperface of $K$ that is the closest to $\mathbf{x}$.\\

\begin{proposition}\label{corr}
Algorithm~\ref{alg} is correct.
\end{proposition}

\begin{proof}
Let us prove by induction on the dimension of $K$ that the algorithm provides us $\mathbf{x}_0\in K$ such that $d(\mathbf{a},\,K)=\|\overrightarrow{\mathbf{a}\mathbf{x}_0}\|$.\\
$\bullet$ If $K$ has dimension $1$, the result is clear.\\
$\bullet$ Now assume that the algorithm is correct for every $(n-1)-$simplex. Let us consider $K$ a $n-$simplex (with $n\geq2$), and prove that the algorithm is correct for $K$. Let $\mathbf{x}$ be the orthogonal projection of $\mathbf{a}$ onto $H$.\\
--- If $\mathbf{x}$ belongs to $K$, then $\mathbf{x}$ is the solution, and the algorithm is correct.\\
--- If $\mathbf{x}$ does not belong to $K$, as $n\geq 2$, we consider the simplex $K'$ defined above, the affine subspace $H'$ and $\mathbf{y}$ the orthogonal projection of $\mathbf{x}$ onto $H'$. By induction hypothesis applied to $\mathbf{y}$ and the $(n-1)-$simplex $K'$, the algorithm provides us $\mathbf{x}_0\in K'$ such that $d(\mathbf{y},\,K')=\|\overrightarrow{\mathbf{y}\mathbf{x}_0}\|$. In particular, $\mathbf{x}_0$ belongs to $K$. Let us now prove that $d(\mathbf{a},\,K)=\|\overrightarrow{\mathbf{a}\mathbf{x}_0}\|$. According to the Pythagorean theorem, as $\overrightarrow{\mathbf{a}\mathbf{x}}$ is orthogonal to $H$, we have
$$d(\mathbf{a},\,K)^2=\|\overrightarrow{\mathbf{a}\mathbf{x}}\|^2+d(\mathbf{x},\,K)^2=\|\overrightarrow{\mathbf{a}\mathbf{x}}\|^2+d(\mathbf{x},\,K')^2,$$
thanks to Lemma~\ref{lemme}.\\
Moreover, as $\overrightarrow{\mathbf{x}\mathbf{y}}$ is orthogonal to $H'$, we have
$$d(\mathbf{x},\,K')^2=\|\overrightarrow{\mathbf{x}\mathbf{y}}\|^2+d(\mathbf{y},\,K')^2=\|\overrightarrow{\mathbf{x}\mathbf{y}}\|^2+\|\overrightarrow{\mathbf{x}_0\mathbf{y}}\|^2=\|\overrightarrow{\mathbf{x}\mathbf{x}_0}\|^2$$
since $\overrightarrow{\mathbf{x}\mathbf{y}}$ is orthogonal to $\overrightarrow{\mathbf{x}_0\mathbf{y}}$.\\
Finally, $d(\mathbf{a},\,K)^2=\|\overrightarrow{\mathbf{a}\mathbf{x}}\|^2+\|\overrightarrow{\mathbf{x}\mathbf{x}_0}\|^2=\|\overrightarrow{\mathbf{a}\mathbf{x}_0}\|^2$ as $\overrightarrow{\mathbf{a}\mathbf{x}}$ is orthogonal to $\overrightarrow{\mathbf{x}\mathbf{x}_0}$, hence $d(\mathbf{a},\,K)=\|\overrightarrow{\mathbf{a}\mathbf{x}_0}\|$ and the algorithm is correct for $K$.
\end{proof}

\begin{remark}\label{Rk}
Let $\mathbf{x}$ be in $H\setminus K$. Then the hyperface of $K$ that is the closest to $\mathbf{x}$ is not necessarily the hyperface of $K$ obtained by suppressing the (or one) vertex of $K$ that is the furthest of $\mathbf{x}$.
\end{remark}

\begin{proof}
Let us consider the following example: let $K$ be the $2-$simplex in the hyperplane $\{x+y+z=1\}$ of $\mathbb{R}^3$. Let us set
$$C=\begin{bmatrix}
0.012 & 0.004 & 0.008\\
0.004 & 0.011 & 0.007\\
0.008 & 0.007 & 0.011
\end{bmatrix},\ \
\mathbf{x}=\begin{bmatrix}
0.470\\
0.534\\
-0.004
\end{bmatrix}.$$
Then the distances from $\mathbf{x}$ to the vertices $e_1$, $e_2$, $e_3$ are respectively $d_1\simeq0.065$, $d_2\simeq0.057$, $d_3\simeq0.062$.\\
Now, the projections of $\mathbf{x}$ onto the edges defined by $\{e_2,\,e_3\}$, $\{e_1,\,e_3\}$, $\{e_1,\,e_2\}$ are respectively
$$p_1=\begin{bmatrix}
0.\\
0.534\\
0.466
\end{bmatrix},\ \
p_2=\begin{bmatrix}
0.47\\
0.\\
0.53
\end{bmatrix},\ \
p_3=\begin{bmatrix}
0.46786667\\
0.53213333\\
0.
\end{bmatrix},$$
and the distances from $\mathbf{x}$ to these edges are $\delta_1\simeq0.039$, $\delta_2\simeq0.048$, $\delta_3\simeq0.0002$.\\
So we have $d_2<d_3<d_1$ but $\delta_3<\delta_1<\delta_2$. Here, the vertex of $K$ that is the furthest of $\mathbf{x}$ is $e_1$, but the distance from $\mathbf{x}$ to $K$ is realized in a point of the edge defined by $\{e_1,\,e_2\}$.
\end{proof}

\subsection{\textbf{Minimisation of $f$ on $K$}}

\noindent Now that we have the algorithm to compute the generalized distance to a standard simplex, it is easy to find the solution that minimises $f$ on $K$: the portfolio that possesses the lowest risk is $P_{\mathbf{x}_0}$, where $\mathbf{x}_0$ is the point of $K$ that realizes the distance from the origin point to $K$, i.e. $d(0,\,K)=\|\mathbf{x}_0\|$.

\section{Application to portfolio optimization}\label{section5}

\noindent Here we determine the portfolio with lowest risk: we determine the convex combination of CAC 40 stocks\footnote{We use the following abbreviations.\\
\begin{tabular}{|l|l|}
  \hline
AC : Accor SA & ACA : Credit Agricole S.A.\\
AI : L'Air Liquide S.A. & AIR : Airbus SE\\
ATO : Atos SE & BN : Danone S.A.\\
BNP : BNP Paribas SA & CA : Carrefour SA\\
CAP : Capgemini SE & CS : AXA SA\\
DG : VINCI SA & DSY : Dassault Systemes SE\\
EL : EssilorLuxottica Societe anonyme & EN : Bouygues SA\\
ENGI : ENGIE SA & FP : TOTAL S.A.\\
GLE : Societe Generale Societe anonyme & HO : Thales S.A.\\
KER : Kering SA & LR : Legrand SA\\
MC : LVMH Moet Hennessy - Louis Vuitton & ML : Cie G\textsuperscript{le} des Et. Michelin\\
MT : ArcelorMittal & OR : L'Oreal S.A.\\
ORA : Orange S.A. & PUB : Publicis Groupe S.A.\\
RI : Pernod Ricard SA & RMS : Hermes International\\
RNO : Renault SA & SAF : Safran SA\\
SAN : Sanofi & SGO : Compagnie de Saint-Gobain S.A.\\
STM : STMicroelectronics N.V. & SU : Schneider Electric\\
SW : Sodexo S.A. & UG : Peugeot S.A.\\
URW : Unibail-Rodamco-Westfield & VIE : Veolia Environnement S.A.\\
VIV : Vivendi & WLN : Worldline\\
  \hline
\end{tabular}} for which the variance is minimal\footnote{For this computation, we do not consider EL, GLE and WLN, for which we don't have enough data.}.\\
We use the mean and the standard deviation of monthly\footnote{The French stock market month (that ends the third Friday in the month) is used.} variation.

\subsection{\textbf{Portfolio optimization from 2007-04-23 to 2020-07-21}}

\noindent Here we consider the period from 2007-04-23 to 2020-07-21, that is to say we start from the highest point of CAC 40 index. Table~\ref{tab:1} gives the mean and the standard deviation of stocks' return rates that appear in the results.

\begin{table}[h]
\caption{Mean and standard deviation of stocks' return rate}
\label{tab:1}
\begin{tabular}{||l||c|c|c|c|c||}
\hline\noalign{\smallskip}
\textrm{Stock}  & \textrm{AI} & \textrm{BN} & \textrm{CA} & \textrm{DSY} & \textrm{ENGI} \\
\noalign{\smallskip}\hline\noalign{\smallskip}
\textrm{Mean}  & 0.73\% & 0.18\% & -0.559\% & 1.53\% & -0.43\% \\
\textrm{Std. Dev.} & 5.73\% & 5.55\% & 8.02\% & 6.58\% & 7.06\% \\
\hline\hline\noalign{\smallskip}
\textrm{Stock}  & \textrm{HO} & \textrm{ORA} & \textrm{RMS} & \textrm{SAN} & \textrm{VIV} \\
\noalign{\smallskip}\hline\noalign{\smallskip}
\textrm{Mean}  & 0.53\% & -0.18\% & 1.6\% & 0.4\% & 0.08\% \\
\textrm{Std. Dev.} & 6.57\% & 6.74\% & 8.63\% & 6.25\% & 6.75\% \\
\noalign{\smallskip}\hline
\end{tabular}
\end{table}

\noindent By using Algorithm~\ref{alg}, we determine the portfolio with allowed short-selling that possesses the lowest risk: this linear combination is given by Table~\ref{tab:2}. The mean of its monthly variation is $0.44\%$ and its standard-deviation $4.35\%$.

\begin{table}[h]
\caption{Portfolio with allowed short-selling that possesses the lowest risk}
\label{tab:2}
\begin{tabular}{||l||c|c|c|c|c||}
\hline\noalign{\smallskip}
\textrm{Stock}  & \textrm{AI} & \textrm{BN} & \textrm{CA} & \textrm{DSY} & \textrm{ENGI} \\
$\mathbf{x}_0$  & 7.28\% & 22.69\% & 3.90\% & 12.63\% & 2.41\% \\
\noalign{\smallskip}\hline\noalign{\smallskip}
\textrm{Stock}  & \textrm{HO} & \textrm{ORA} & \textrm{RMS} & \textrm{SAN} & \textrm{VIV} \\
$\mathbf{x}_0$  & 3.81\% & 21.05\% & 8.96\% & 16.57\% & 0.70\% \\
\noalign{\smallskip}\hline
\end{tabular}
\end{table}

\noindent We observe here that the linear combination obtained is a convex combination, which means that this portfolio is also the portfolio without short-selling that possesses the lowest risk. In geometrical terms, this means that the orthogonal projection of the origin point onto the hyperplane $H$ already belongs to the simplex $K$.

\begin{table}[h]
\caption{Yearly return rate of the portfolio}
\label{tab:3}
\begin{small}
\begin{tabular}{||l||c|c|c|c|c|c|c||}
\hline\noalign{\smallskip}
\textrm{Year}  & 2007 & 2008 & 2008 & 2010 & 2011 & 2012 & 2013 \\
\noalign{\smallskip}\hline\noalign{\smallskip}
\textrm{Porfolio's\ r.\ r.}  & 1.85\% & -23.04\% & 5.09\% & 12.84\% & -0.80\% & 14.67\% & 7.81\% \\
\textrm{CAC\ 40's\ r.\ r.}  & -5.12\% & -43.87\% & 22.87\% & 0.34\% & -23.45\% & 23.54\% & 14.72\% \\
\hline\hline\noalign{\smallskip}
\textrm{Year}  & 2014 & 2015 & 2016 & 2017 & 2018 & 2019 & 2020 \\
\noalign{\smallskip}\hline\noalign{\smallskip}
\textrm{Porfolio's\ r.\ r.}  & 12.89\% & 16.16\% & 3.95\% & 14.38\% & 1.47\% & 33.12\% & 3.37\% \\
\textrm{CAC\ 40's\ r.\ r.}  & 0.93\% & 7.30\% & 5.64\% & 12.40\% & -14.65\% & 30.33\% & -15.34\% \\
\noalign{\smallskip}\hline
\end{tabular}
\end{small}
\end{table}

\noindent Table~\ref{tab:3} and Figure~\ref{fig:2} give the yearly return rate of this portfolio. We notice that the portfolio is more profitable and more regular than the index: indeed, its mean is $7.41\%$ and its standard deviation $11.93\%$, whereas for the index, the mean is $1.12\%$ and the standard deviation $19.61\%$. Moreover, the return rate of the portfolio is almost never negative.

\begin{figure}[h]
\includegraphics[height=6cm]{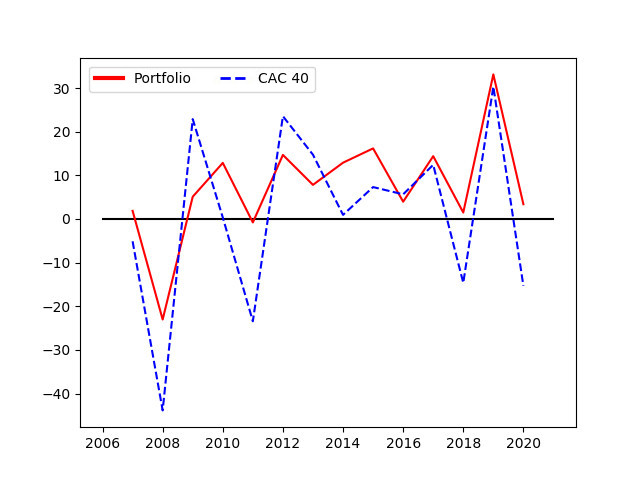}
\caption{Yearly return rate of the portfolio}
\label{fig:2}
\end{figure}

\noindent Finally, Table~\ref{tab:4} and Figure~\ref{fig:3} give the final value of the portfolio compared with the CAC 40 index and show the monthly variations of their value.

\begin{figure}[h]
\includegraphics[height=6cm]{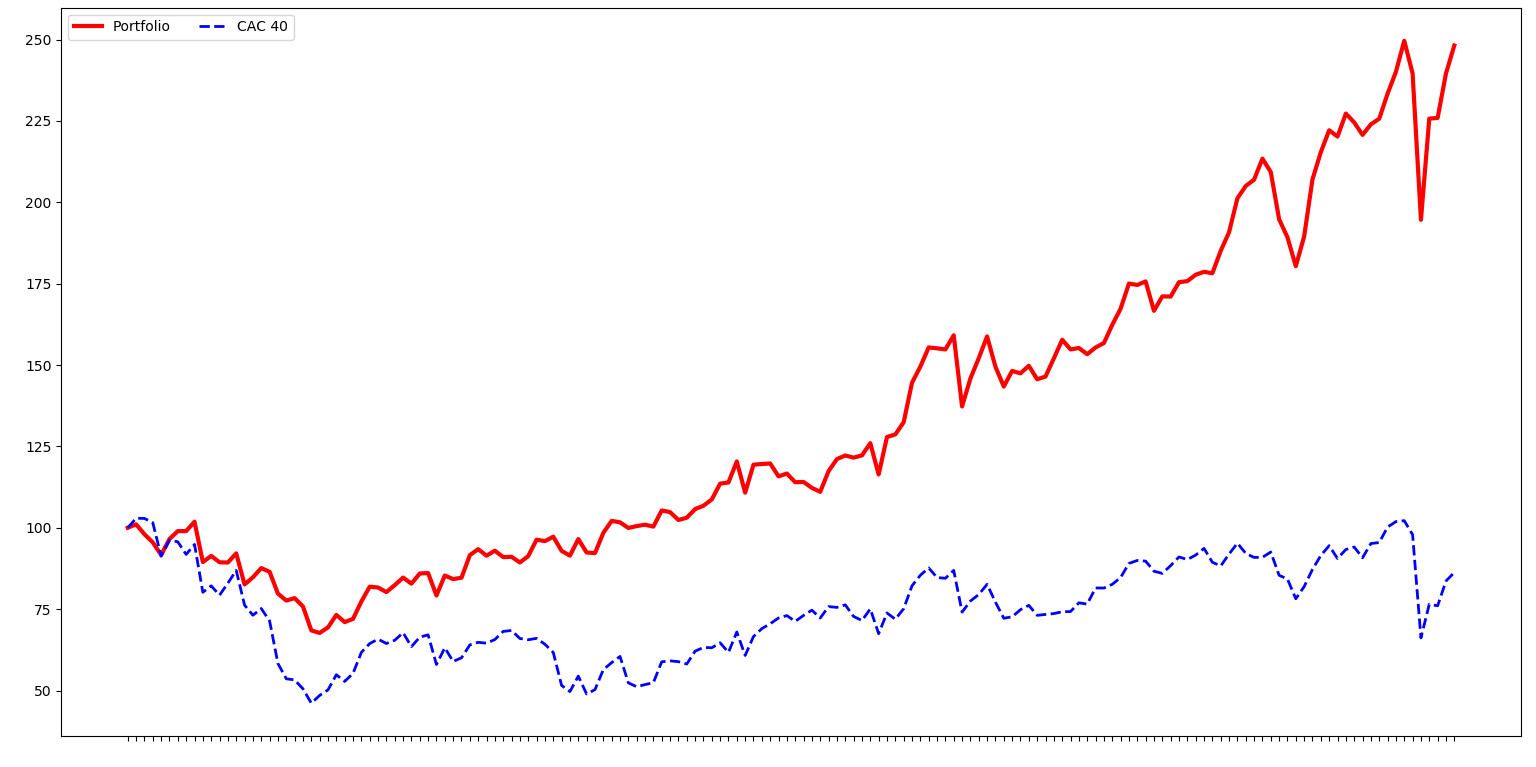}
\caption{Portfolio optimization from 2007-04-23 to 2020-07-21}
\label{fig:3}
\end{figure}

\begin{table}[h]
\caption{Final value of the portfolio compared with the CAC 40 index}
\label{tab:4}
\begin{tabular}{||l||c|c||}
\hline\noalign{\smallskip}
\textrm{Date}  & \textrm{2009-01-19} & \textrm{2020-07-21}\\
\noalign{\smallskip}\hline\noalign{\smallskip}
\textrm{Portfolio}  & 100 & \rou{248.16} \\
\textrm{CAC\ 40}  & 100 & \rou{86.26} \\
\noalign{\smallskip}\hline
\end{tabular}
\end{table}

\subsection{\textbf{Portfolio optimization from 2009-01-19 to 2020-07-21}}

\noindent Here we consider the period from 2009-01-19 to 2020-07-21, that is to say we start from the lowest point of CAC 40 index. As in previous section, Table~\ref{tab:5} gives the mean and the standard deviation of stocks' return rates that appear in the results.\\

\begin{table}[h]
\caption{Mean and standard deviation of stocks' return rate}
\label{tab:5}
\begin{tabular}{||l||c|c|c|c|c|c||}
\hline\noalign{\smallskip}
\textrm{Stock}  & \textrm{AI} & \textrm{BN} & \textrm{CA} & \textrm{DSY} & \textrm{ENGI} & \textrm{HO} \\
\noalign{\smallskip}\hline\noalign{\smallskip}
\textrm{Mean}  & 1.03\% & 0.44\% & -0.13\% & 1.93\% & -0.470\% & 0.79\% \\
\textrm{Std. Dev.} & 5.36\% & 5.34\% & 7.9699\% & 6.0\% & 7.0900\% & 6.4799\%\\
\hline\hline\noalign{\smallskip}
\textrm{Stock}  & \textrm{OR} & \textrm{ORA} & \textrm{RI} & \textrm{RMS} & \textrm{SAN} & \textrm{VIV} \\
\noalign{\smallskip}\hline\noalign{\smallskip}
\textrm{Mean}  & 1.3599\% & -0.18\% & 1.0\% & 1.8599\% & 0.64\% & 0.37\% \\
\textrm{Std. Dev.} & 5.42\% & 6.77\% & 5.93\% & 7.4399\% & 6.03\% & 6.77\% \\
\noalign{\smallskip}\hline
\end{tabular}
\end{table}

\noindent The portfolio with allowed short-selling that possesses the lowest risk is given by the following linear combination in Table~\ref{tab:6}. The mean of its monthly variation is $0.76\%$ and its standard-deviation $4.17\%$.\\

\begin{table}[h]
\caption{Portfolio with allowed short-selling that possesses the lowest risk}
\label{tab:6}
\begin{tabular}{||l||c|c|c|c|c|c||}
\hline\noalign{\smallskip}
\textrm{Stock}  & \textrm{AI} & \textrm{BN} & \textrm{CA} & \textrm{DSY} & \textrm{ENGI} & \textrm{HO} \\
$\mathbf{x}_0$  & 10.70\% & 28.25\% & 5.42\% & 20.91\% & -2.34\% & 1.88\% \\
\noalign{\smallskip}\hline\noalign{\smallskip}
\textrm{Stock}  & \textrm{OR} & \textrm{ORA} & \textrm{RI} & \textrm{RMS} & \textrm{SAN} & \textrm{VIV} \\
$\mathbf{x}_0$  & -20.66\% & 19.75\% & -3.07\% & 17.26\% & 18.36\% & 3.54\% \\
\noalign{\smallskip}\hline
\end{tabular}
\end{table}

\noindent Now, according to Algorithm~\ref{alg}, the portfolio without short-selling that possesses the lowest risk is given by the following convex combination in Table~\ref{tab:7}. The mean of its monthly variation is $0.87\%$ and its standard-deviation $4.22\%$.\\

\begin{table}[h]
\caption{Portfolio without allowed short-selling that possesses the lowest risk}
\label{tab:7}
\begin{tabular}{||l||c|c|c|c|c||}
\hline\noalign{\smallskip}
\textrm{Stock}  & \textrm{AI} & \textrm{BN} & \textrm{CA} & \textrm{DSY}  & \textrm{HO}   \\
$\mathbf{x}_0$  & 6.17\% & 19.22\% & 3.43\% & 17.37\% & 2.96\% \\
\noalign{\smallskip}\hline\noalign{\smallskip}
\textrm{Stock}  & \textrm{ORA} & \textrm{RMS} & \textrm{SAN} & \textrm{VIV} \\
$\mathbf{x}_0$  & 17.15\% & 15.60\% & 16.73\% & 1.37\% \\
\noalign{\smallskip}\hline
\end{tabular}
\end{table}

\noindent Let us note that some stocks that were in the linear combination of the portfolio with allowed short-selling now disappear: the geometric explanation of this fact immediately comes from Lemma~\ref{lemme}.\\

\noindent The yearly return rate is given by Table~\ref{tab:8}. Here again, as shown by Table~\ref{tab:8} and Figure~\ref{fig:4}, the portfolio is more profitable and more regular than the index: its mean is $13.27\%$ and its standard deviation $11.99\%$, whereas for the index, the mean is $5.94\%$ and the standard deviation $16.78\%$. Moreover, the return rate of the portfolio is negative for only two years, and the absolute value of these negative return rates very small.\\

\begin{table}[h]
\caption{Yearly return rate of the portfolio}
\label{tab:8}
\begin{tabular}{||l||c|c|c|c|c|c||}
\hline\noalign{\smallskip}
\textrm{Year}  & 2009 & 2010 & 2011 & 2012 & 2013 & 2014 \\
\noalign{\smallskip}\hline\noalign{\smallskip}
\textrm{Portfolio's\ r.\ r.}  & 17.95\% & 15.82\% & -1.49\% & 27.80\% & 5.76\% & 7.87\% \\
\textrm{CAC\ 40's\ r.\ r.}  & 29.51\% & 0.34\% & -23.45\% & 23.54\% & 14.72\% & 0.93\% \\
\hline\hline\noalign{\smallskip}
\textrm{Year}  & 2015 & 2016 & 2017 & 2018 & 2019 & 2020 \\
\noalign{\smallskip}\hline\noalign{\smallskip}
\textrm{Portfolio's\ r.\ r.}  & 23.81\% & -2.25\% & 18.18\% & 2.86\% & 38.27\% & 4.72\% \\
\textrm{CAC\ 40's\ r.\ r.}  & 7.30\% & 5.64\% & 12.40\% & -14.65\% & 30.33\% & -15.34\% \\
\noalign{\smallskip}\hline
\end{tabular}
\end{table}

\begin{figure}[h]
\includegraphics[height=6cm]{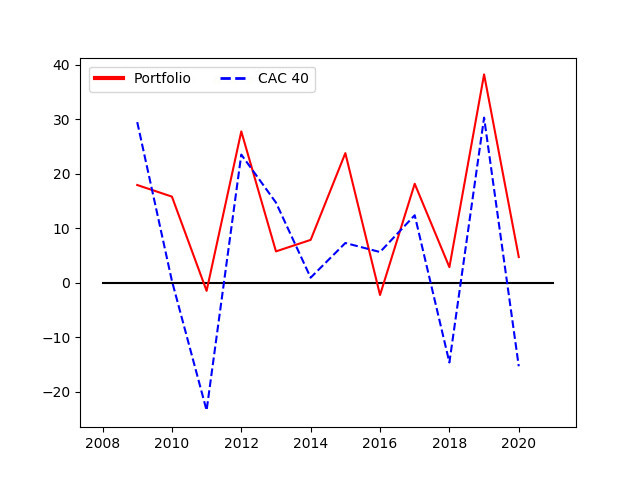}
\caption{Yearly return rate of the portfolio}
\label{fig:4}
\end{figure}

\noindent Finally, Table~\ref{tab:9} and Figure~\ref{fig:5} above give the final value of the portfolio compared with the CAC 40 index and show the monthly variations of their value.

\begin{table}
\caption{Final value of the portfolio compared with the CAC 40 index}
\label{tab:9}
\begin{tabular}{||l||c|c||}
\hline\noalign{\smallskip}
\textrm{Date}  & \textrm{2009-01-19} & \textrm{2020-07-21}\\
\noalign{\smallskip}\hline\noalign{\smallskip}
\textrm{Portfolio}  & 100 & \rou{417.97} \\
\textrm{CAC\ 40}  & 100 & \rou{170.73} \\
\noalign{\smallskip}\hline
\end{tabular}
\end{table}

\begin{figure}
\includegraphics[height=6cm]{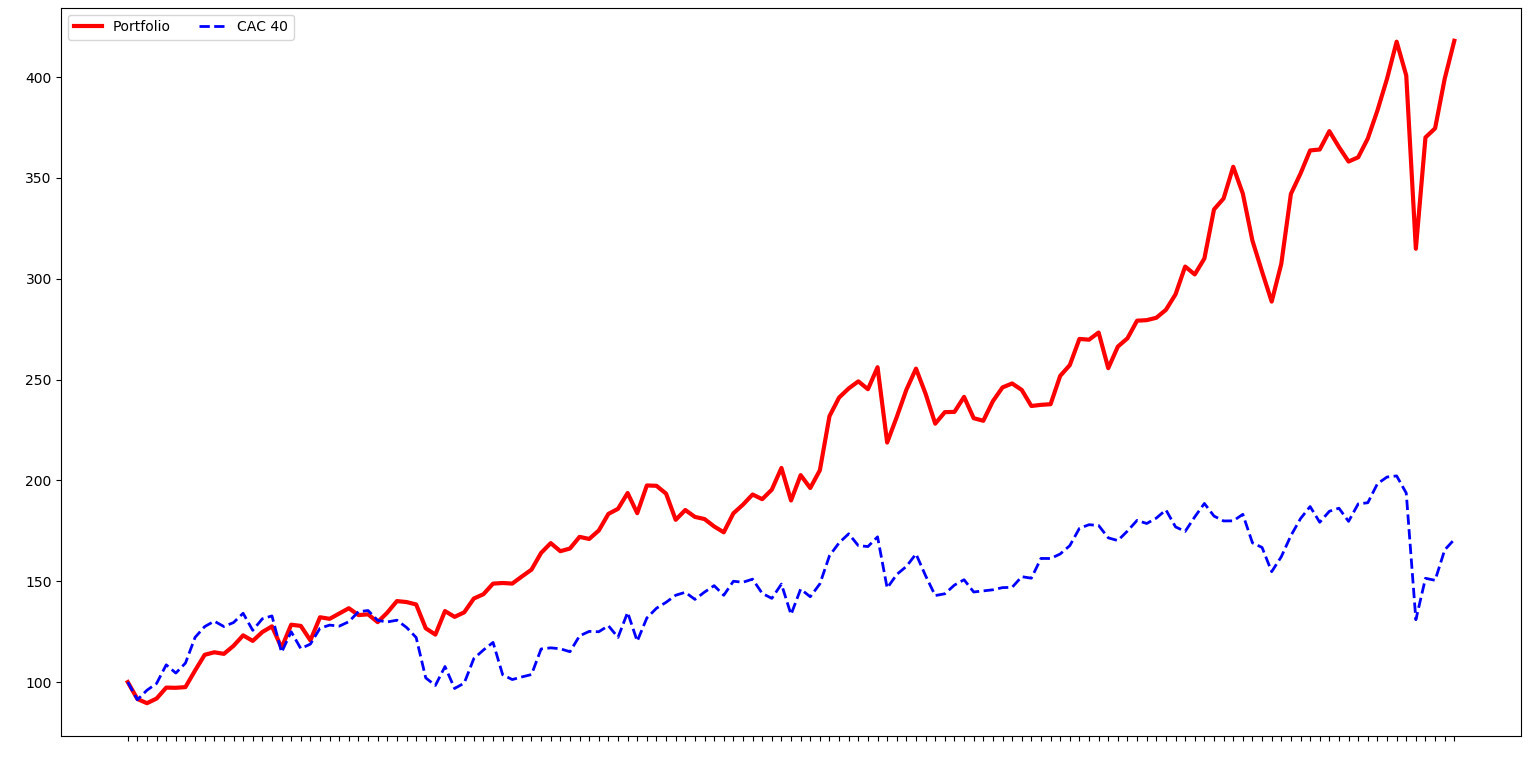}
\caption{Portfolio optimization from 2009-01-19 to 2020-07-21}
\label{fig:5}
\end{figure}

\section*{Appendix --- Python programs}

\noindent Here we give a possible way to program Algorithm~\ref{alg} in Python as well as the subroutine used to compute an orthogonal projection.\\

\noindent The function \verb"orth_proj(c,a,J)" computes an orthogonal projection, where\\
--- \verb"c" is the covariance matrix,\\
--- \verb"a" is the point of which we want to compute the orthogonal projection,\\
--- \verb"J" is the list of indices of $p$ vectors of $\mathcal{E}$ that define the affine subspace onto which we want to project \verb"a".\\
\verb"def orth_proj(c,a,J):"\\
\verb"    p=len(J); n=len(c); i0=J[0]; L=list(set(range(n))-set(J))"\\
\verb"    #L is the complementary of J"\\
\verb"    #Matrix of the system"\\
\verb"    Mpart1=np.array([[c[i,j]-c[i0,j] for j in J] for i in J[1:]])"\\
\verb"    Mpart2=np.ones((1,p))"\\
\verb"    M=np.concatenate((Mpart1,Mpart2),axis=0)"\\
\verb"    #Second part of the system"\\
\verb"    b=np.array([sum(a[j]*(c[i,j]-c[i0,j]) for j in range(n))"\\
\verb"        for i in J[1:]]+[1])"\\
\verb"    #Solving"\\
\verb"    sol=np.linalg.solve(M,b)"\\
\verb"    x=np.zeros(n); x[J]=sol"\\
\verb"    return x"\\

\noindent The function \verb"mini_dist_fct(c,a)" finds the point that realizes the minimal distance from $\mathbf{a}$ to the standard $(n-1)-$simplex and also returns the square of this distance: it is a possible version of Algorithm~\ref{alg}.\\
\verb"def mini_dist_fct(c,a):"\\
\verb"    #Scalar product"\\
\verb"    def phi(x,y):"\\
\verb"        return np.dot(x,np.dot(c,y))"\\
\verb"    #Canonical basis"\\
\verb"    n=len(c)"\\
\verb"    e=[np.array(j*[0]+[1]+(n-j-1)*[0]) for j in range(n)]"\\
\verb"    dico={}"\\
\verb"    #Recursive function mini_dist(c,a,J)"\\
\verb"    def mini_dist(c,a,J):"\\
\verb"        #Orthogonal projection of a"\\
\verb"        x=orth_proj(c,a,J)"\\
\verb"        #Case 1: the orthogonal projection belongs to the simplex"\\
\verb"        if all(t>=0 for t in x):"\\
\verb"            return [x,phi(x-a,x-a)]"\\
\verb"        #Case 2: the orthogonal projection does not belong to the"\\
\verb"        #simplex and the simplex has dimension 1"\\
\verb"        elif len(J)==2:"\\
\verb"            d0=phi(x-e[J[0]],x-e[J[0]])"\\
\verb"            d1=phi(x-e[J[1]],x-e[J[1]])"\\
\verb"            if d0<=d1:"\\
\verb"                return [e[J[0]],phi(e[J[0]]-a,e[J[0]]-a)]"\\
\verb"            else:"\\
\verb"                return [e[J[1]],phi(e[J[1]]-a,e[J[1]]-a)]"\\
\verb"        #Case 3: the orthogonal projection does not belong to the"\\
\verb"        #simplex and the simplex has dimension greater than 1"\\
\verb"        else:"\\
\verb"            #Looking for the hyperface that is the closest to x"\\
\verb"            s=J[0]"\\
\verb"            if str(set(J)-{s})+str(x) in dico:"\\
\verb"                delta=dico[str(set(J)-{s})+str(x)]"\\
\verb"            else:"\\
\verb"                delta=mini_dist(c,x,list(set(J)-{s}))"\\
\verb"                dico[str(set(J)-{s})+str(x)]=delta"\\
\verb"            d=delta[1]"\\
\verb"            for j in J[1:]:"\\
\verb"                if str(set(J)-{j})+str(x) in dico:"\\
\verb"                    delta0=dico[str(set(J)-{j})+str(x)]"\\
\verb"                else:"\\
\verb"                    delta0=mini_dist(c,x,list(set(J)-{j}))"\\
\verb"                    dico[str(set(J)-{j})+str(x)]=delta0"\\
\verb"                d0=delta0[1]"\\
\verb"                if d0<d:"\\
\verb"                    s=j; d=d0"\\
\verb"            #Projection onto the simplex defined by the"\\
\verb"            #closest hyperface"\\
\verb"            J=list(set(J)-{s})"\\
\verb"            if str(set(J))+str(x) in dico:"\\
\verb"                delta=dico[str(set(J))+str(x)]"\\
\verb"            else:"\\
\verb"                delta=mini_dist(c,x,J)"\\
\verb"                dico[str(set(J))+str(x)]=delta"\\
\verb"            x=delta[0]"\\
\verb"            return [x,phi(x-a,x-a)]"\\
\verb"    #Computing the point realizing the minimal distance"\\
\verb"    return mini_dist(c,a,list(range(n)))"

\end{document}